\documentclass[11pt]{llncs}

\usepackage[utf8]{inputenc}
\usepackage[T1]{fontenc}
\usepackage[english]{babel}
\usepackage{amsmath,amssymb,amsfonts,array}
\usepackage{url}
\usepackage{algorithm}
\usepackage{algorithmic}
\usepackage{tikz-cd}
\usepackage{lmodern}
\usepackage{microtype}
\usepackage{enumerate}
\usepackage{booktabs}
\usepackage{cite}
\usepackage{varioref}
\usepackage[pdfpagelabels=true]{hyperref}
\usepackage{cleveref}
\usepackage{todonotes}
\usepackage{pgfplots}

\usetikzlibrary{decorations.text}

\hypersetup{
	linktoc = page,
	pdfpagemode = UseNone,
	colorlinks,
	linkcolor={red!50!black},
	citecolor={blue!50!black},
	urlcolor={blue!80!black}
}

\newcommand{\xor}[0]{\oplus}

\newcommand{\size}[1]{|#1|}

\begin{document}
\title{A Number-Theoretic Error-Correcting Code}

\author{
  		Eric Brier\inst{1}, Jean-S\'{e}bastien Coron\inst{2}, R\'{e}mi G\'{e}raud\inst{1, 3},
		Diana Maimu\c{t}\inst{3}, and David Naccache\inst{2,3}
}

	\institute{
		Ingenico\\
		28-32 boulevard de Grenelle, {\sc f}-75015, Paris, France\\
		\url{{surname.name}@ingenico.com}
		\and
		Université du Luxembourg\\
		6 rue Richard Coudenhove-Kalergi, 1359 Luxembourg, Luxembourg\\
		\url{{surname.name}@uni.lu}
		\and
		\'Ecole normale sup\'erieure\\
		D\'epartement d'Informatique\\
		45 rue d'Ulm, {\sc f}-75230, Paris {\sc cedex 05}, France\\
		\url{{surname.name}@ens.fr}
	}

\maketitle

\begin{abstract}
In this paper we describe a new error-correcting code (ECC) inspired by the Naccache-Stern cryptosystem. 
While by far less efficient than Turbo codes, the proposed ECC happens to be more efficient than some established ECCs for certain sets of parameters.

The new ECC adds an appendix to the message. The appendix is the modular product of small primes representing the message bits. The receiver
recomputes the product and detects transmission errors using modular division and lattice reduction.
\end{abstract}

\section{Introduction}

Error-correcting codes (ECCs) are essential to ensure reliable communication. ECCs work by adding redundancy which enables detecting and correcting mistakes in received data.
This extra information is, of course, costly and it is important to keep it to a minimum: there is a trade-off between how much data is added for error correction purposes (bandwidth), and the number of errors that can be corrected (correction capacity).

Shannon showed \cite{Shannon48} in 1948 that it is in theory possible to encode messages with a minimal number of extra bits\footnote{Shannon's theorem states that the best achievable expansion rate is $1 - H_2(p_b)$, where $H_2$ is binary entropy and $p_b$ is the acceptable error rate.}.
Two years later, Hamming \cite{Hamming50} proposed a construction inspired by parity codes, which provided both error detection and 
error correction. Subsequent research saw the emergence of more efficient codes, such as Reed-Muller \cite{Muller54,Reed54}
and Reed-Solomon \cite{ReedS60}. The latest were generalized by Goppa \cite{Goppa81}. These codes are known as algebraic-geometric codes.

Convolutional codes were first presented in 1955 \cite{elias55}, while recursive systematic convolutional codes \cite{turbo93} were introduced in 1991. Turbo codes \cite{turbo93} were indeed revolutionary, given their closeness to the channel capacity (``near Shannon limit'').

\paragraph{Results:} This paper presents a new error-correcting code, as well as a form of message size improvement based on the hybrid use of two ECCs one of which is inspired by the Naccache-Stern (NS) cryptosystem \cite{NaccacheS97,CMNS08a}.
For some codes and parameter choices, the resulting hybrid codes outperform the two underlying ECCs.

The proposed ECC is unusual because it is based on number theory rather than on binary operations.

\section{Preliminaries}

\subsection{Notations}

Let $\mathfrak P = \left\{ p_1 = 2, \dotsc \right\}$ be the ordered set of prime numbers. 
Let $\gamma \geq 2$ be an encoding base. For any $m \in \mathbb N$ (the ``message''), let $\{m_i\}$ be the digits of $m$ in base $\gamma$ \emph{i.e.}:
\begin{equation*}
m = \sum_{i=0}^{k-1} \gamma^i m_i \qquad m_i \in [0, \gamma - 1], \quad k = \lceil \log_\gamma m \rceil 
\end{equation*}
We denote by $h(x)$ the Hamming weight of $x$, \emph{i.e.} the sum of $x$'s digits in base~2, and, by $\size{y}$ the bit-length of $y$.

\subsection{Error-Correcting Codes}

Let $\mathcal M = \{0,1\}^k$ be the set of messages, $\mathcal C = \{0,1\}^{n}$ the set of encoded messages.
Let $\mathcal P$ be a parameter set.

\begin{definition}[Error-Correcting Code]
	An \emph{error-correcting code} is a couple of algorithms:
	\begin{itemize}
		\item An algorithm $\mu$, taking as input some message $m\in \mathcal M$, as
		well as some public parameters $\mathsf{params} \in \mathcal P$, and outputting $c \in \mathcal C$.
		\item An algorithm $\mu^{-1}$, taking as input $\tilde c \in \mathcal C$ as well as parameters $\mathsf{params} \in \mathcal P$, and outputting $m \in \mathcal M \cup \{\bot\}$.
		
		The $\bot$ symbol indicates that decoding failed.
	\end{itemize}
\end{definition}

\begin{definition}[Correction Capacity]
	Let $(\mu, \mu^{-1}, \mathcal M, \mathcal C, \mathcal P)$ be an error-correcting code.
	There exists an integer $t \geq 0$ and some parameters $\mathsf{params} \in \mathcal P$ such that, for all $e \in \mathcal \{0,1\}^{n}$ such that $h(e) \leq t$, 
	\begin{equation*}
	\mu^{-1} \left( \mu \left( m , \mathsf{params} \right) \oplus e, \mathsf{params} \right) = m,
	\qquad 
	\forall m \in \mathcal M
	\end{equation*}
and for all $e$ such that $h(e) > t$,
	\begin{equation*}
	\mu^{-1} \left( \mu \left( m , \mathsf{params} \right) \oplus e, \mathsf{params} \right) \neq m,
	\qquad 
	\forall m \in \mathcal M.
	\end{equation*}
	$t$ is called the \emph{correction capacity} of $(\mu, \mu^{-1}, \mathcal M, \mathcal C, \mathcal P)$.
\end{definition}

\begin{definition}
\label{nktcode}
A code of message length $k$, of codeword length $n$ and with a correction capacity $t$ is called an $(n,k,t)$-code. The ratio $\rho = \frac{n}{k}$ is called the code's \emph{expansion rate}.
\end{definition}

\section{A New Error-Correcting Code}
\label{sec:NewEcc}

Consider in this section an existing $(n, k, t)$-code $C = (\mu, \mu^{-1}, \mathcal M, \mathcal C, \mathcal P)$.
For instance $C$ can be a Reed-Muller code.
We describe how the new  $(n',k,t)$-code $C' = (\nu, \nu^{-1}, \mathcal M, \mathcal C', \mathcal P')$ is constructed.
 
\paragraph{Parameter Generation:}
To correct $t$ errors in a $k$-bit message, we generate a prime $p$ such that:
\begin{equation}
\label{eq:equp}
2 \cdot p_{k}^{2t} < p < 4 \cdot p_{k}^{2t}
\end{equation}
As we will later see, the size of $p$ is obtained by bounding the worst case in which all errors affect the end of the message. $p$ is a part of $\mathcal P'$.

\paragraph{Encoding:}
Assume we wish to transmit a $k$-bit message $m$ over a noisy channel.
Let $\gamma = 2$ so that $m_i$ denote the $i$-th bit of $m$, and define: 
\begin{equation}
\label{defcm}
c(m): = \prod_{i=1}^k p_i^{m_i} \bmod{p}
\end{equation}

The integer generated by \Cref{defcm} is encoded using $C$ to yield $\mu(c(m))$.
Finally, the encoded message $\nu(m)$ transmitted over the
noisy channel is defined as:
\begin{equation}
\mu(m):=m \| \mu(c(m))
\end{equation}

Note that, if we were to use $C$ directly, we would have encoded $m$ (and not $c$).
The value $c$ is, in most practical situations, much shorter than $m$. As is explained in \Cref{sec:perfNS}, $c$ is smaller than $m$ (except the cases in which $m$ is very small and which are not interesting in practice) and thereby requires fewer extra bits for correction.
For appropriate parameter choices, this provides a more efficient encoding, as compared to $C$.

\paragraph{Decoding:}
Let $\alpha$ be the received\footnote{\emph{i.e.} encoded and potentially corrupted} message.
Assume that at most $t$ errors occurred during transmission:
$$ \alpha =  \nu(m) \xor e  =  m' \| (\mu(c(m)) \xor e') $$
where the error vector $e$ is such that $h(e)=h(m' \xor m)+h(e') \leq t$. 

Since $c(m)$ is encoded with a $t$-error-capacity code, we can recover the correct value of $c(m)$ from $\mu(c(m))
\xor e'$ and compute the quantity:
\begin{equation}
s=\frac{c(m')}{c(m)} \bmod{p}
\end{equation}
Using \Cref{defcm} $s$ can be written as:
\begin{equation}
s = \frac{a}{b} \bmod{p}, \quad
\begin{cases}
a & = \prod\limits_{(m'_i=1) \wedge (m_i=0)} p_i \\
b & = \prod\limits_{(m'_i=0) \wedge (m_i=1)} p_i
\end{cases}
\end{equation}

Note that since $h(m' \xor m) \leq t$, we have that $a$ and $b$
are strictly smaller than $(p_k)^t$. \Cref{theo} from
\cite{cryptorational} shows that given $t$ the receiver can recover $a$ and
$b$ efficiently using a variant of Gauss' algorithm
\cite{vallee}.

\begin{theorem}
\label{theo} Let $a,b \in {\mathbb Z}$ such that $-A \leq a \leq
A$ and $0<b \leq B$. Let $p$ be some prime integer such that
$2AB<p$. Let $s=a \cdot b^{-1} \mod p$. Then given $A$, $B$, $s$
and $p$, $a$ and $b$ can be recovered in polynomial time.
\end{theorem}

As $0 \leq a \leq A$ and $0 <b \leq B$ where $A=B=(p_k)^t-1$
and  $2AB<p$ from \Cref{eq:equp}, we can recover $a$ and $b$ from
$t$ in polynomial time. Then, by testing the divisibility of $a$
and $b$ with respect to the small primes $p_i$, the receiver can recover
$m' \xor m$ and eventually $m$.

A numerical example is given in \Cref{sec:ToyEx}.

\paragraph{Bootstrapping:}

Note that instead of using an existing code as a
sub-contractor for protecting $c(m)$, the sender may also recursively
apply the new scheme described above. To do so consider $c(m)$ as a message, and protect $\overline{c} = c(c(\cdots c(c(m)))$, which is a rather small value, against
accidental alteration by replicating it $2t+1$ times. The receiver will use a majority vote to detect the errors in $\overline{c}$. 

\subsection{Performance of the New Error-Correcting Code for $\gamma=2$}
\label{sec:perfNS}

\begin{lemma}
	\label{lem:1}
	The bit-size of $c(m)$ is:
	\begin{equation}
	\label{sizep}
	\log_2 p \simeq {2 \cdot t}\log_2(k \ln k).
	\end{equation}
\end{lemma}

\begin{proof}
	From \Cref{eq:equp} and the Prime Number Theorem\footnote{$p_k \simeq k\ln k$.}.
	\qed
\end{proof}

The total output length of the new error-correcting code is therefore $\log_2 p$, plus the length $k$ of the message $m$.

$C'$ outperforms the initial error correcting code $C$ if, for equal error capacity $t$ and message length $k$, it outputs
a shorter encoding, which happens if $n' < n$, keeping in mind that both $n$ and $n'$ depend on $k$.

\begin{corollary}
	\label{corollary}
	Assume that there exists a constant $\delta>1$ such that, for $k$ large enough, $n(k) \geq \delta k$.
	Then for $k$ large enough, $n'(k) \leq n(k)$.
\end{corollary}

\begin{proof}
	Let $k$ be the size of $m$ and $k'$ be the size of $c(m)$.
	
	\noindent
	We have $n'(k) = k + n(k')$, therefore
	\begin{align*}
	n(k) - n'(k) = n(k) - (k + n(k')) \geq (\delta-1) k - n(k').
	\end{align*}
	Now,
	\begin{align*} 
	(\delta-1) k - n(k') \geq 0 \Leftrightarrow (\delta-1) k \geq n(k'). 
	\end{align*}
	But $n(k') \geq \delta k'$,
	hence 
	\begin{align*}
	(\delta-1)k \geq \delta k' \Rightarrow k \geq \displaystyle \frac{k' \delta}{(\delta-1)}.
	\end{align*}
	Finally, from \Cref{lem:1}, $k' = O(\ln \ln k!)$, which guarantees that there exists a value of $k$ above which $n'(k) \leq n(k)$.
	\qed
\end{proof}

In other terms, any correcting code whose encoded message size is growing linearly with message size can benefit from the described construction.

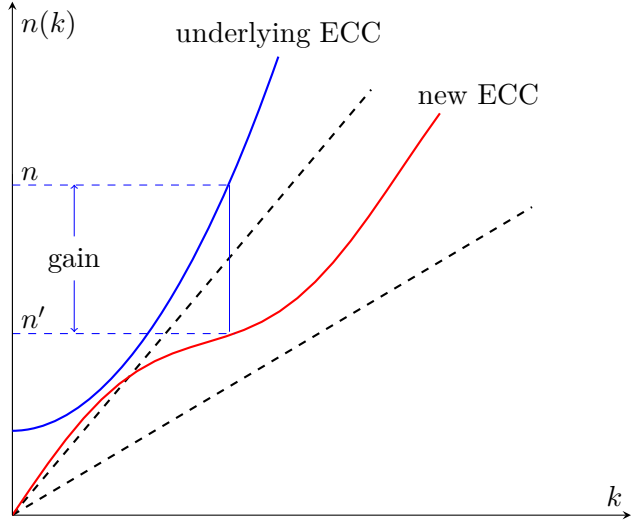
\begin{figure}[ht]
\begin{center}
\begin{tikzpicture}[fdesc/.style={anchor=north east,sloped,font=\normalsize,pos=#1},fdesc/.default=1.02]
\begin{axis}[
		scale=1.2,
        axis x line=middle, 
        axis y line=middle, 
        xmin=0,xmax=10,xlabel=\normalsize{$k$},
        ymin=0,ymax=14, ylabel=\normalsize{$n(k)$},
        tick label style={color=white},
        ticks = none
        ]
    \addplot[domain=0:8.4, black, thick, dashed] {x}
          node[below,fdesc] {lower bound $n=k$};
    \addplot[domain=0:5.8, black, thick, dashed] {2*x}
          node[below,fdesc] {bound on underlying ECC $n=\delta k$};
          
    \addplot[domain=0:2.2*pi, red, thick] {1.5*x+sin(deg(x))))}
    	node[above, color=black] {\normalsize{~~~~~~~~new ECC}};
    \addplot[domain=0:18, white, very thick] {0.5*x^2+2.3};
    \addplot[domain=0:3.6, blue, thin, dashed] {4.95};
    
    \draw (0,0) -- (0,5.4) node[right] {\normalsize{$n'$}};
    \draw (0,0) -- (0,9.3) node[right] {\normalsize{$n$}};
    
    \addplot[domain=0:3.55, blue, ultra thin, dashed] {9};
  
    \draw [blue, thick] (0,2.3) parabola bend (0,2.3) (43,12.5)
    	node[above, color=black] {\normalsize{underlying ECC}};
        
    \addplot[blue,ultra thin] table{
    A B
    3.51 5
    3.51 9
    };
    
     \addplot[blue,ultra thin,->] table{
    A B
    1 7.6
    1 9
    };
    \addplot[blue,ultra thin,<-] table{
    A B
    1 5
    1 6.4
    }
    	node[above, color=black] {\normalsize{gain}};
\end{axis}
\end{tikzpicture}
\end{center}
\caption{Illustration of \Cref{corollary}. For large enough values of $k$, the new ECC uses smaller codewords as compared to the underlying ECC.}
\label{fig:cor}
\end{figure}

\paragraph{Expansion Rate:}
Let $k$ be the length of $m$ and consider the bit-size of the corresponding codeword as in \Cref{sizep}. The expansion rate $\rho$ is: 

\begin{equation}
\label{expansionrate11}
\rho = \displaystyle \frac{\size{m \| \mu(c(m))}}{\size{m}} = \displaystyle \frac{k + \size{\mu(c(m))}}{k} = 1 + \frac{\size{\mu(c(m))}}{k}
\end{equation}

\subsubsection{Reed-Muller Codes}
\label{sec:perfRM}

We illustrate the idea with Reed-Muller codes.
Reed-Muller (R-M) codes are a family of linear codes. Let $r \geq 0$ be an integer, and
$N = \log_2 n$, it can apply to messages of size
\begin{equation}
k=\sum_{i=1}^r \binom{N}{i}
\end{equation}

Such a code can correct up to $t=2^{N-r-1}-1$ errors. Some examples of $\{n,k,t\}$ triples
are given in \Cref{trm}. 
For instance, a message of size $163$
bits can be encoded as a $256$-bit string, among which up to $7$
errors can be corrected.

\begin{table}
	{\tt
		\begin{center}
			\begin{tabular}{|c||r|r|r|r|r|r|r|r|r|r|} \hline
				~~$n$ ~~& 16 & 64 & 128 & 256 & 512   & 2048 & 8192 & 32768 &
				131072 \\ \hline $k$ & 11 & 42 & 99 & 163  &  382  &  1024 & 5812
				& 9949 & 65536 \\ \hline $t$ & 1 & 3 & 3 & 7 & 7  &31 & 31 & 255 &
				255\\ \hline
			\end{tabular}
		\end{center}}
		\caption{Examples of length $n$, dimension $k$, and error capacity $t$ for Reed-Muller code.} \label{trm}
	\end{table}
	
	To illustrate the benefit of our approach, consider a 5812-bit message, which we wish to protect against up to 31 errors.
	
	A direct use of Reed-Muller would require $n(5812) = 8192$ bits as seen in \Cref{trm}. Contrast this with
	our code, which only has to protect $c(m)$, that is 931 bits as shown by \Cref{sizep},
	yielding a total size of $5812 + n(931) = 5812 + 2048 = 7860$ bits.
	
	Other parameters for the Reed-Muller primitive are illustrated in \Cref{tnc}.
	\begin{table}
		{\tt
			\begin{center}
				\begin{tabular}{|c||r|r|r|} \hline
					~~$n'$ ~~& ~638~  & ~7860~ &  ~98304~ \\ 
					\hline $k$ & ~382~ & ~5812~ & ~65536~  \\ 
					\hline $c(m)$ & ~157~ & ~931~  & ~9931~ \\ 
					\hline ~${\rm RM}(c(m))$~ & ~256~ & ~2048~ & ~32768~ \\ \hline $t$ & ~7~ & ~31~ & ~255~
					\\ \hline
				\end{tabular}
			\end{center}}
			\caption{$(n, k, t)$-codes generated from Reed-Muller by our construction.} \label{tnc}
		\end{table}
		
		\Cref{tnc} shows that for large message sizes and a small
		number of errors, our error-correcting code slightly outperforms Reed-Muller code. 

\subsection{The case $\gamma > 2$}

The difficulty in the case $\gamma > 2$ stems from the fact that a binary error in a $\gamma$-base message
will in essence scramble all digits preceding the error. As an example,
\begin{equation*}
\underline{12200210}12202012010011120202_3 + 2^{30} = \underline{12200210}22112000112220110110_3
\end{equation*}
Hence, unless $\gamma = 2^\Gamma$ for some $\Gamma$, a generalization makes sense only for channels
over which transmission uses $\gamma$ symbols. In such cases, we have the following: a $k$-bit message
$m$ is pre-encoded as a $\gamma$-base $\kappa$-symbol message $m'$. Here $\kappa = \lceil k/\log_2 \gamma \rceil$. 
\Cref{eq:equp} becomes:
\begin{equation*}
2 \cdot p_\kappa^{2t(\gamma - 1)} < p < 4 \cdot p_\kappa^{2t(\gamma - 1)}
\end{equation*}
Comparison with the binary case is complicated by the fact that here $t$ refers to the number of
\emph{any} errors regardless their semiologic meaning. In other words, an error transforming a $0$ into a $2$
counts exactly as an error transforming $0$ into a $1$.

\begin{example}
	\label{ex:gamma3}
As a typical example, for $t=7$, $\kappa = 10^6$ and $\gamma = 3$, $p_\kappa = 15485863$ and $p$ is a $690$-bit number.

For the sake of comparison, $t = 7$, $k = 1584963$ (corresponding to $\kappa = 10^6$) and $\gamma = 2$, yield $p_k = 25325609$ and a $346$-bit $p$.
\end{example}

\section{Improvement Using Smaller Primes}

The construction described in the previous section can be improved by choosing a smaller prime $p$, but comes at a price; namely decoding
becomes only heuristic.

\paragraph{Parameter Generation:} The idea consists in generating a prime $p$ smaller than
before. Namely, we generate a $p$ satisfying~:
\begin{equation}
\label{eqnewp}
2^u \cdot p_{k}^t < p < 2^{u+1} \cdot p_{k}^{ t}
\end{equation}
for some small integer $u \geq 1$. 

\paragraph{Encoding and Decoding:} Encoding remains as
previously. The redundancy $c(m)$ being approximately half as
small as the previous section's one, we have~:
\begin{equation}
\label{eqnab}
s=\frac{a}{b} \mod p, \quad 
\begin{cases}
a & =   \prod\limits_{\substack{(m'_i=1) \wedge (m_i=0)}} p_i \\
b & =  \prod\limits_{\substack{(m'_i=0) \wedge (m_i=1)}} p_i
\end{cases}
\end{equation}
and since there are at most $t$ errors, we must have~:
\begin{equation}
\label{eqab}
a \cdot b \leq (p_k)^t
\end{equation}
We define a finite sequence $\{A_i,B_i\}$ of integers such that
$A_i=2^{u \cdot i}$ and $B_i=\lfloor 2p/A_i \rfloor$. From
\Cref{eqnewp,eqab} there must be at least one index $i$ such that $0 \leq a \leq A_i$ and $0 <b \leq B_i$.
Then using \Cref{theo}, given $A_i$,
$B_i$, $p$ and $s$, the receiver can recover $a$ and $b$, and eventually
$m$. 

The problem with that approach is that we lost the guarantee that $\{a,b\}$ is unique. 
Namely we may find another $\{a',b'\}$ satisfying \Cref{eqnab} for some other
index $i'$. We expect this to happen with negligible probability
for large enough $u$, but this makes the modified code heuristic
(while perfectly implementable for all practical purposes).

\subsection{Performance}

\begin{lemma}
	The bit-size of $c(m)$ is:
	\begin{equation}
	\label{sizesp}
	\log_2 p \simeq u + t \log_2(k \ln k).
	\end{equation}
\end{lemma}

\begin{proof}
	Using \Cref{eqnewp} and the Prime Number Theorem.
	\qed
\end{proof}

Thus, the smaller prime variant has a shorter $c(m)$.

As $u$ is a small integer (\emph{e.g.} $u=50$), it follows immediately from \Cref{eq:equp} that, for large $n$ and $t$, the size of the new prime $p$ will be approximately half the size of the prime $p$ generated in the preceding section. 

This brings down the minimum message size $k$ above which our construction provides an improvement over the bare
underlying correcting code.

\paragraph{Note:} In the case of Reed-Muller codes, this variant provides no improvement over the technique described in \Cref{sec:NewEcc} for the following reasons: (1) by design, Reed-Muller codewords
are powers of 2; and (2) \Cref{sizesp} cannot yield a twofold reduction in $p$. Therefore we cannot hope to reduce $p$ enough to get a smaller codeword.

That doesn't preclude other codes to show benefits, but the authors did not look for such codes.

\section{Prime Packing Encoding}

It is interesting to see whether the optimization technique of \cite{CMNS08a} yields more efficient ECCs.
Recall that in \cite{CMNS08a}, the $p_i$s are distributed amongst $\kappa$ packs. Information is encoded by picking
one $p_i$ per pack. This has an immediate impact on decoding: when an error occurs and a symbol $\sigma$ is replaced
by a symbol $\sigma'$, both the numerator and the denominator of $s$ are affected by \emph{additional} prime factors. 

Let $C = (\mu, \mu^{-1},\mathcal M, \mathcal C, \mathcal P)$ be a $t$-error capacity code, such that it is possible to efficiently recover $c$ from $\mu(c) \oplus e$ for any $c$ and any $e$, where $h(e) \leq t$.
Let $\gamma \geq 2$ be a positive integer.

Before we proceed, we define $\kappa := \lceil k / \log_2\gamma \rceil$ and
\begin{equation*}
f := f(\gamma, \kappa, t) = \prod_{i=k-t}^{k} p_{\gamma i}.
\end{equation*}

\paragraph{Parameter Generation:}
 Let $p$ be a prime number such that: 
\begin{equation}
\label{eqpcompress}
2 \cdot f^{2} < p < 4 \cdot f^{2}
\end{equation}
Let $\mathcal{\hat C} = \mathcal M \times \mathbb Z_p$ and
$\mathcal{\hat P} = (\mathcal P \cup \mathfrak P) \times \mathbb N$.
We now construct a variant of the ECC presented in \Cref{sec:NewEcc} from $C$ and denote it 
\begin{equation*}
\hat C = \left(\nu, \nu^{-1}, \mathcal M, \mathcal{\hat C}, \mathcal{\hat P} \right).
\end{equation*}

\paragraph{Encoding:}

We define the ``redundancy'' of a $k$-bit message $m \in \mathcal M$ (represented as $\kappa$ digits in base $\gamma$) by:
\begin{equation*}
\hat c(m) := \prod_{i=0}^{\kappa-1} p_{i\gamma + m_i + 1} \bmod p
\end{equation*}
A message $m$ is encoded as follows:
\begin{equation*}
\nu(m) :=  m \| \mu\left( \hat c\left(m\right)\right)
\end{equation*}

\paragraph{Decoding:}

The received information $\alpha$ differs from $\nu(m)$ by a certain number of bits. Again, we 
assume that the number of these differing bits is at most $t$. Therefore $\alpha = \nu(m) \oplus e$, where $h(e) \leq t$.
Write $e = e_m \| e_{\hat c} $ such that
\begin{equation*}
\alpha= \nu(m) \oplus e = m \oplus e_m \| \mu(\hat c(m)) \oplus e_{\hat c} = m' \|  \mu(\hat c(m)) \oplus e_{\hat c}.
\end{equation*}

Since $h(e) = h(e_m) + h(e_{\hat c}) \leq t$, the receiver can recover efficiently $\hat c(m)$ from $\alpha$. It is then possible to compute
\begin{align*}
s & := \frac{\hat c(m')}{\hat c(m)} \bmod p = \frac{\displaystyle\prod_{i=0}^{\kappa-1}p_{i\gamma+ m'_i + 1}}{\displaystyle\prod_{i=0}^{\kappa-1}p_{i\gamma+m_i + 1}} \bmod p.
\end{align*}
\begin{equation}
s = \frac{a}{b} \bmod{p}, \quad
\begin{cases}
a & = \displaystyle\prod \limits_{\substack{m'_i \neq m_i}} p_{i\gamma+ m'_i + 1} \\
b & = \displaystyle\prod \limits_{\substack{m_i \neq m'_i}} p_{i\gamma+ m_i + 1}
\end{cases}
\end{equation}

As $h(e) = h(e_m) + h(e_{\hat c}) \leq t$, we have that $a$ and $b$ are strictly smaller than $f(\gamma,\kappa)^{2t}$. As $A = B = f(\gamma,\kappa)^{2t}-1$, we observe from \Cref{eqpcompress} that $2AB < p$. 
We are now able to recover
$a, b$, $\operatorname{gcd}(a,b)=1$ such that $s = a/b \bmod p$ using lattice reduction \cite{vallee}.

Testing the divisibility of $a$ and $b$ by $p_1, \dotsc, p_{\kappa \gamma}$ the receiver can recover $e_m = m' \oplus m$, and from that get $m = m' \oplus e_m$. Note that by construction only one prime amongst $\gamma$ is used per 
``pack'': the receiver can therefore skip on average $\gamma/2$ primes in the divisibility testing phase.

\subsection{Performance}

Rosser's theorem \cite{dusart1999k,rosser} states that for $ n \geq 6$,
\begin{equation*}
\ln n + \ln\ln n - 1 < \frac{p_n}{n} < \ln n + \ln \ln n
\end{equation*}
\emph{i.e.} $p_n < n(\ln n + \ln \ln n )$. Hence a crude upper bound of $p$ is
	\begin{align*}
	p & < 4f(\kappa,\gamma, t)^2 \\
	& = 4 \left( \prod_{i=\kappa-t}^{\kappa} p_{\gamma i} \right)^2 \\
	& \leq 4 \prod_{i = \kappa-t}^\kappa \left( i\gamma (\ln{i\gamma} + \ln\ln(i\gamma)) \right)^2 \\
	& \leq 4 \gamma^{2t} \left( \frac{\kappa!}{(\kappa-t-1)!}\right)^2 \left( \ln \kappa\gamma + \ln\ln \kappa\gamma \right)^{2t}
	\end{align*}

Again, the total output length of the new error-correcting code is $n' = k + |p|$.

Plugging $\gamma = 3$, $\kappa = 10^6$ and $t = 7$ into \Cref{eqpcompress} we get a $410$-bit~$p$. This improves over Example~\ref{ex:gamma3} where
$p$ was $690$ bits long.

\bibliographystyle{splncs03}
\bibliography{Bibliography}

\begin{thebibliography}{10}
\providecommand{\url}[1]{\texttt{#1}}
\providecommand{\urlprefix}{URL }

\bibitem{turbo93}
Berrou, C., Glavieux, A., Thitimajshima, P.: {Near Shannon Limit
  Error-Correcting Coding and Decoding: Turbo-Codes}. In: IEEE International
  Conference on Communications - ICC'93. vol.~2, pp. 1064--1070 (May 1993)

\bibitem{CMNS08a}
Chevallier-Mames, B., Naccache, D., Stern, J.: {Linear Bandwidth Naccache-Stern
  Encryption}. In: Proceedings of the 6th International Conference on Security
  and Cryptography for Networks - SCN '08. pp. 327--339. Springer-Verlag (2008)

\bibitem{dusart1999k}
Dusart, P.: The $k^{{\rm th}}$ prime is greater than $k (\ln k+ \ln \ln k-1)$
  for $k \geq 2$. Mathematics of Computation pp. 411--415 (1999)

\bibitem{elias55}
Elias, P.: {Coding for Noisy Channels}. In: IRE Conv. Rec. pp. 37--46 (1955)

\bibitem{cryptorational}
Fouque, P.A., Stern, J., Wackers, J.G.: {CryptoComputing with Rationals}. In:
  Proceedings of the 6th International Conference - Financial Cryptography'02.
  Lecture Notes in Computer Science, vol. 2357, pp. 136--146. Springer (2002)

\bibitem{Goppa81}
Goppa, V.D.: {Codes on Algebraic Curves}. Soviet Math. Doklady  24,  170--172
  (1981)

\bibitem{Hamming50}
Hamming, R.W.: {Error Detecting and Error Correcting Codes}. Bell System
  Technical Journal  29(2),  147--160 (1950)

\bibitem{Muller54}
Muller, D.E.: {Application of Boolean Algebra to Switching Circuit Design and
  to Error Detection}. IRE Transactions on Information Theory (3),  6--12
  (1954)

\bibitem{NaccacheS97}
Naccache, D., Stern, J.: A new public-key cryptosystem. In: Advances in
  Cryptology - EUROCRYPT'97. Lecture Notes in Computer Science, vol. 1233, pp.
  27--36. Springer (1997)

\bibitem{Reed54}
Reed, I.: {A Class of Multiple-Error-Correcting Codes and the Decoding Scheme}.
  IRE Transactions on Information Theory (4),  38--49 (Sep 1954)

\bibitem{ReedS60}
Reed, I.S., Solomon, G.: Polynomial {C}odes {O}ver {C}ertain {F}inite {F}ields.
  Journal of the Society for Industrial and Applied Mathematics  8(2),
  300--304 (1960)

\bibitem{rosser}
Rosser, J.B.: {The $n$-th Prime is Greater than $n \ln n$}. In: Proceedings of
  the London Mathematical Society. vol.~45, pp. 21--44 (1938)

\bibitem{Shannon48}
Shannon, C.: {A Mathematical Theory of Communication}. Bell System Technical
  Journal  27,  379--423, 623--656 (1948)

\bibitem{vallee}
Vall{\'{e}}e, B.: {Gauss' Algorithm Revisited}. J. Algorithms  12(4),  556--572
  (1991)

\end{thebibliography}

\appendix

\section{Toy Example}
\label{sec:ToyEx}

\label{ex1}
\noindent
Let $m$ be the 10-bit message $1100100111.$
\noindent
For $t=2$, we let $p$ be the smallest prime number greater than $2 \cdot 29^4$, {\it i.e.} $p=707293$.
\noindent
We generate the redundancy:
$$c(m)=2^1 \cdot 3^1\cdot 5^0 \cdot 7^0 \cdot 11^1 \cdot 13^0 \cdot 17^0 \cdot 19^1 \cdot 23^1 \cdot 29^1 \bmod 707293 $$
$$\Rightarrow c(m)= 836418 \bmod 707293 = 129125.$$
As we focus on the new error-correcting code we simply omit the Reed-Muller component.
The encoded message is 
$$\nu(m)={\tt 1100100111_2} \| {\tt 129125_{10}}.$$
Let the received encoded message be 
$\alpha={\tt 1100101011_2} \| {\tt 129125_{10}}$.
\noindent
Thus,
$$c(m')=2^1 \cdot 3^1 \cdot 5^0 \cdot 7^0 \cdot 11^1 \cdot 13^0 \cdot 17^1 \cdot 19^0 \cdot 23^1 \cdot 29^1 \bmod p$$
$$\Rightarrow c(m')= 748374 \bmod 707293 = 41081.$$
Dividing by $c(m)$ we get
$$s = \frac{c(m')}{c(m)} = \frac{41081}{129125} \bmod 707293 = 632842$$
Applying the rationalize and factor technique we obtain $s = \displaystyle \frac{17}{19} \bmod 707293$.
\noindent
It follows that $m' \oplus m = {\tt 0000001100}$. Flipping the bits retrieved by this calculation, we recover $m$.

\end{document}